\documentclass[11pt,reqno]{amsart}%
\usepackage[latin1]{inputenc}
\usepackage{mathrsfs}
\usepackage{amsmath}
\usepackage{amssymb}
\usepackage{amsthm}
\usepackage{amsfonts}
\usepackage{amstext}
\usepackage{amsopn}
\usepackage{amsxtra}
\usepackage{mathrsfs}
\usepackage{graphicx}
\usepackage{cancel}

\newtheorem{theorem}{Theorem}
\newtheorem{proposition}[theorem]{Proposition}
\newtheorem{lemma}[theorem]{Lemma}

\newcommand\R{{\ensuremath {\mathbb R} }}

\renewcommand\phi{\varphi}

\newcommand{\gF}{\mathfrak{F}}

\newcommand{\cH}{\mathcal{H}}

\newcommand{\eps}{\epsilon}

\newcommand{\bx}{\mathbf{x}}

\newcommand{\by}{\mathbf{y}}

\newcommand{\bp}{\mathbf{p}}
\newcommand{\bP}{\mathbf{P}}
\newcommand{\bq}{\mathbf{q}}
\newcommand{\bk}{\mathbf{k}}
\newcommand{\bK}{\mathbf{K}}

\renewcommand{\epsilon}{\varepsilon}

\renewcommand{\geq}{\geqslant}
\renewcommand{\leq}{\leqslant}

\renewcommand{\tilde}{\widetilde}

\newcommand{\dk}{{\rm d}\bk}

\newcommand{\ud}{\mathrm{d}}
\newcommand{\ue}{\mathrm{e}}
\newcommand{\ui}{\mathrm{i}}

\newcommand{\ren}{\mathrm{ren}}

\newcommand{\uppar}[1]{\ensuremath{^{(#1)}}}

\date{\today}

\begin{document}

\title{The Renormalised Bogoliubov-Fröhlich Hamiltonian}

\author{Jonas Lampart}
\email{jonas.lampart@u-bourgogne.fr}
\address{CNRS \& Laboratoire Interdisciplinaire Carnot de Bourgogne, UMR 6303 CNRS \& Universit\'{e} de Bourgogne Franche-Comt\'{e}, 9 Av. A. Savary, 21078 Dijon CEDEX, France}

\begin{abstract}
The Bogoliubov-Fröhlich Hamiltonian models the interaction of an impurity with the excitations of a Bose-Einstein condensate. 
It has been observed that the dependence of the ground state energy on the ultraviolet cutoff differs significantly from what would be expected from similar well-known models.
We give a detailed explanation of this UV behaviour, and provide an explicit representation of the renormalised Hamiltonian.
\end{abstract}

\maketitle

\section{The Bogoliubov-Fröhlich Model}

\subsection{Introduction}

The Bogoliubov-Fröhlich Hamiltonian models the dynamics of impurities immersed in a Bose-Einstein Condensate. The interaction between the impurities and the bosons can lead to the formation of polarons.
Such effects are well-known from similar models, such as H. Fröhlich's  paradigmatic model of electrons interacting with optical phonons in  a crystal~\cite{Hfrohlich1954}, and have been studied in various contexts, in both physics and mathematics.

In a BEC, the excitations of particles out of the joint condensate wavefunction can be described by a quantum field, similar to phonons.
In the Bogoliubov approximation, the dispersion of the excitations is related to Bogoliubov's famous formula for the excitation spectrum that helps explain superfluidity~\cite{bogoliubov1947}.
The interaction with the impurities is essentially a contact interaction, due to the dilute nature of the condensate gas. 
If the interaction is approximated by a linear coupling of impurities to the excitation field, one obtains the 
Bogoliubov-Fröhlich Hamiltonian, whose form closely resembles the Hamiltonian used by Fröhlich.
However, the contact interaction gives rise  to much stronger ultraviolet singularities. In this article we will show how these can be addressed.

Polarons in a BEC have recently attracted considerable attention, and a variety of theoretical methods  has been employed in their study~\cite{ardila2018, cucchietti2006, tempere2009, casteels2011,Levinson2015, Vlietinck2015, christensen2015, grusdt2015, grusdt2016, grusdt2016d, grusdt2016, grusdt2016f, grusdt2017,   Kain2016,  ichmoukhamedov2019, drescher2019, mistakidis2019}. 
Comparing numerical results, Grusdt~\cite{grusdt2016f} found that the system is well described by 
a Bogoliubov-Fröhlich Hamiltonian when the interaction is either attractive or not too strong.
Bose condensates with mobile impurities have recently been realised in experiments~\cite{Catani2012, Joergensen2016, Hu2016, camargo2018}, providing first tests of these theoretical predictions.


Concerning the ultraviolet behaviour, it was observed in~\cite{grusdt2015} that the ground state energy displays a divergence of the form $e_1 \Lambda + e_2 \log \Lambda$, where $\Lambda$ is the UV-cutoff. The linearly divergent term was expected, and the mechanism behind it is well understood from other models (see~\ref{sect:UV Polaron} for a detailed discussion). The logarithmic divergence is specific to the model and its origin is more subtle.

The goal of this article is to explain this logarithmic divergence in detail and give an explicit non-perturbative description of the 
renormalised Bogoliubov-Fröhlich Hamiltonian. This operator can be obtained by subtracting suitable numbers from the 
Hamiltonian with UV-cutoff and taking the limit $\Lambda\to \infty$. Such a limiting procedure often fails to specify 
the renormalised Hamiltonian in an explicit way. In practice, one is thus forced to 
work with the cutoff model, taking a sequence of larger and larger values for $\Lambda$ until one observes convergence of the relevant quantities.
However, we are able to present a formula for the renormalised operator, Eq.~\eqref{eq:Hren}, that can be studied directly.
This operator is self-adjoint, and thus generates a unitary time-evolution, and its spectrum is bounded from below, as we will show in Section~\ref{sect:Hren}.
We also characterise the domain of definition of the operator, consisting of those vectors $\Psi$ in the Hilbert space for which $\| H_\ren \Psi\|^2=\langle \Psi| H_\ren^2 \Psi \rangle <\infty$.
The explicit description of $H_\ren$ will open up new possibilities of studying aspects of the model analytically.
The renormalised Hamiltonian specifies the model without UV-cutoff in a non-perturbative way. 
Perturbation theory can then be applied to this model without having to deal with divergences.
Furthermore, in Section~\ref{sect:pert} we present an expansion of $H_\ren$ in powers of the coupling constant $g$ that could prove useful for both analytical and numerical approaches.

The characterisation of $H_\ren$ and its domain can be understood in the position representation using 
interior-boundary-conditions. These boundary conditions for quantum field theories were proposed in~\cite{TeTu15,TeTu16} 
(see there for additional references). They resemble the well-known Bethe-Peierls conditions for contact 
pseudo-potentials, but in a context where the particle number is not conserved.
Reflecting the creation and annihilation of excitation-particles on contact, they relate the singular behaviour of $\Psi$ near collision configurations to the values of $\Psi$ with fewer bosons. We explain how this relates to the Bogoliubov-Fröhlich model in Section~\ref{sect:IBC}.

\subsection{The Hamiltonian}

In this section we introduce the relevant objects and notation, for a detailed introduction of the Bogoliubov-Fröhlich model we refer to~\cite{grusdt2016c}. 

For the most part, we will restrict our discussion to the model with a single impurity, in order to keep the notation simple. We comment on the case of multiple impurities in Section~\ref{sect:mult}.
The Hamiltonian (with ultraviolet cutoff $\Lambda$) is given by
\begin{equation}\label{eq:H_L}
 H_\Lambda = H_0 + g \Phi_\Lambda(\bx),
\end{equation}
where
\begin{equation}
 H_0=\frac{\bp^2}{2M} + \ud \Gamma(\omega)= \frac{\bp^2}{2M} + \int \dk \omega(k)a^\dagger_\bk a_\bk
\end{equation}
is the sum of the kinetic energies of impurity and bosons, and the cutoff interaction is given by
\begin{equation}
 \Phi_\Lambda(\bx) = \int  \dk v_\Lambda(k) \ue^{\ui \bk\bx}\left( a_\bk + a^\dagger_{-\bk}\right).
\end{equation}
The dispersion relation of the excitations is given by
\begin{equation}
 \omega(k)=c k \sqrt{1+(k\xi)^2/2}
\end{equation}
and the cutoff form factor is
\begin{equation}
 v_\Lambda(k)=\left\lbrace
 \begin{aligned}
&\frac{1}{(2\pi)^{3/2}} \left(\frac{(\xi k)^2}{2+(\xi k)^2} \right)^{1/4} \qquad &k\leq \Lambda\\
&0                  &k>\Lambda.
\end{aligned}
\right.
\end{equation}
Here, $g$ is a coupling constant, $M$ is the mass of the impurity, $\bp$ its momentum, and 
and $\bx=-\ui \nabla_{\bp}$ its position operator. The constants the $\xi$ and $c$ are the coherence length and the speed of sound in the BEC, respectively.
The constant $g$ is an effective coupling constant that may be related to the physical parameters of the BEC via 
\begin{equation}
 g= \frac{2 \pi \sqrt \rho a}{\mu},
\end{equation}
where $\rho$ is the condensate density, $a$ the scattering length of the impurity-boson interaction and $\mu$ the reduced mass (see~\cite[Sect.3]{grusdt2016c} for a discussion).

For fixed boson number $n$, we will denote the function
\begin{equation}\label{eq:H_0 fct}
 H_0(\bp;\bk_1,\dots,\bk_n)=H_0(\bp;\bK)=\frac{\bp^2}{2M} + \sum_{i=1}^n \omega(k_i),
\end{equation}
where the semicolon delimits the momentum variables of the impurity (or, in Section~\ref{sect:mult}, the impurities) and the bosons.

For a finite cutoff, the expression~\eqref{eq:H_L} defines a self-adjoint operator on the domain $D(H_\Lambda)=D(H_0)$ of vectors $\Psi$ in the Hilbert space $\cH=\cH_\mathrm{I}\otimes \gF$, the tensor product of the space $\cH_\mathrm{I}=L^2(\R^3)$ for the impurity and the bosonic Fock space $\gF$ of the excitations, for which $H_0\Psi$ has finite norm.

\begin{proposition}
 For $\Lambda<\infty$, $H_\Lambda$ is self-adjoint on the domain 
 \begin{equation}
  D(H_\Lambda)=D(H_0)=\left\{ \Psi \in \cH \vert H_0\Psi \in \cH \right\}
 \end{equation}
 and bounded from below.
\end{proposition}
\begin{proof}
 Using~\cite[Prop.3.6]{Der2003}, we find that $H_\Lambda$ is bounded relative to $H_0^{1/2}$. Since $H_0^{1/2}\leq \eps H_0 + (4\eps)^{-1}$, the claim then follows from the Kato-Rellich Theorem~\cite[Thm.X.12]{ReSi2}.  
\end{proof}

The model is of course translation-invariant, so total momentum is conserved. One can represent the Hamiltonian as a function of the total momentum $\bP$ by applying the unitary transformation $U=\ue^{\ui \bx \ud \Gamma(\bk)}$, which maps the impurity momentum $\bp$ into the total momentum $\bP=\bp+ \ud \Gamma(\bk)$.
The transformed wavefunction $U\Psi$ can be considered as a function of the total momentum taking values in $\gF$. Then
\begin{equation}
 \langle \Phi | H_\Lambda \Psi \rangle = \int \ud \bP \langle (U\Phi)(\bP)| H_\Lambda (\bP)(U\Psi)(\bP)\rangle_\gF,
\end{equation}
with the transformed Hamiltonian at total momentum $\bP$
\begin{equation}
 H_\Lambda (\bP)= H_0(\bP) + g\Phi_\Lambda,
\end{equation}
where
\begin{equation}
 H_0(\bP)= \frac1{2M}\left(\bP-\ud \Gamma(\bk)\right)^2  + \ud\Gamma(\omega)
\end{equation}
and $\Phi_\Lambda=\Phi_\Lambda(0)$ (this notation should not be confused with the function~\eqref{eq:H_0 fct}).

We will mostly study the translation-invariant model without fixing total momentum. This is convenient for the discussion of the UV-behavior, since one can simply fix the total momentum after renormalisation.

\section{Renormalisation of the Bogoliubov-Fröhlich Hamiltonian}

Having defined the Hamiltonian, we may immediately observe two instances of the UV-problem. First of all, the norm
\begin{equation}
 \|H_\Lambda(\bP)|\emptyset\rangle_{\bP}\|^2= \frac{P^2}{2M}+g^2 \int \ud \bk |v_\Lambda(k)|^2\sim \frac{g^2}{6 \pi^2}\Lambda^3,
\end{equation}
diverges, so the vacuum (at total momentum $\bP$) cannot be in the domain of definition of $H_\infty(\bP)$ (if the latter exists).
Secondly, the perturbative expression for the ground state energy,
\begin{equation}\label{eq:E_0 pert}
 \frac{P^2}{2M} - g^2\int \ud \bk\frac{|v_\Lambda(k)|^2}{\tfrac{1}{2M}(\bP-\bk)^2+\omega(k)} + g^3(\dots)
\end{equation}
is also divergent, indicating that the ground state energy of $H_\infty(\bP)$ would have to be minus infinity (though one should not infer any quantitative predictions from this formula, as the non-interacting ground state does not have a mass-gap). 

We sketch in Section~\ref{sect:UV Polaron} how these problems can be addressed for general Hamiltonians of Fröhlich type. 
In our case, there will be an additional logarithmically divergent expression, and the procedure for constructing $H_\ren$ will be given in detail in Sect.~\ref{sect:Hren}.
Our main result is
\begin{theorem}\label{thm:ren}
 There exists a self-adjoint and lower bounded operator \\$(H_\ren, D(H_\ren))$, and numbers $E_\Lambda$, $\Lambda>0$, such that 
 \begin{equation*}
  H_\Lambda - E_\Lambda \to H_\ren.
 \end{equation*}
\end{theorem}
The content of this theorem is proved in Propositions~\ref{prop:Hren},~\ref{prop:conv} below (see there for the precise statement on convergence). Importantly, Theorem~\ref{thm:ren} proves that the ground state energy of $H_\Lambda-E_\Lambda$ remains finite, since $H_\ren$ is bounded below and the spectrum converges. It also shows that the unitary propagators converge,
\begin{equation}
 \ue^{\ui E_\Lambda t}\ue^{-\ui H_\Lambda t}\to \ue^{-\ui H_\ren t},
\end{equation}
in norm. For a detailed discussion of the notion of resolvent convergence see~\cite[Sect.VIII.7]{ReSi1}.

We emphasise that the objects whose existence is stated here are in fact explicit. 
The Hamiltonian $H_\ren$ and its domain $D(H_\ren)$ are given by~\eqref{eq:Hren} and~\eqref{eq:Dren}, respectively. Its interpretation in terms of interior-boundary conditions is discussed in Sect.~\ref{sect:IBC}. The numbers $E_\Lambda=E_\Lambda\uppar{1}+E_\Lambda\uppar{2}$, with $E_\Lambda\uppar{1}\sim e_1\Lambda$, $E_\Lambda\uppar{2}\sim e_2\log\Lambda$,  depending on the parameters $M,g,\xi,c$ of the model, are given in~\eqref{eq:E_0 pert det} and~\eqref{eq:E(2)}. These agree with formal perturbation theory up to order $g^4$, so our result shows that it is sufficient to subtract these two perturbative terms to fully renormalise the Hamiltonian. 

The renormalised operators at total momentum $\bP$, $H_\ren(\bP)$ can simply be obtained by applying the unitary transformation $U$ to the renormalised operator $H_\ren$ and then fixing $\bP$. in the expicit representation of $H_\ren$ this amounts to simply replacing $\bp$ by $\bP- \sum {\bk_j}$ (where the sum is taken over the appropriate range).  
Since the renormalisation of $H_\Lambda$ is achieved by simply subtracting the constant $E_\Lambda$ the operators $H_\Lambda(\bP)$ are renormalised by subtracting \emph{the same constant}. In particular, differences of ground state energies $E_\Lambda(\bP)$ of $H_\Lambda(\bP)$ at $\bP\neq \bP'$ will stay finite as $\Lambda\to \infty$. This suggests that the effective mass of the polaron, i.e.,
\begin{equation}
 M_{\mathrm{eff},\Lambda}= \left(\frac{\partial^2 E_\Lambda(\bP)}{\partial P_1^2} \right)^{-1}
\end{equation}
will also have a finite limit as $\Lambda \to \infty$. However, our results do not imply this since we do not show that the second derivative of $E_\infty(\bP)$ is finite and does not vanish.

\subsection{UV-Behaviour of Fröhlich-Type Models}\label{sect:UV Polaron}
\leavevmode\\
Hamiltonians of the Fröhlich type have been studied for various different form factors $v(k)$ and dispersion relations $\omega(k)$. With some restrictions, they provide a class of models where the UV-problem can be solved rigorously on the level of the Hamiltonian. Let us briefly review the picture of their UV-behaviour that emerges and cite some key results (assuming that $\omega\geq m$ to focus on the UV problem):

\begin{enumerate}
 \item If $\int \ud \bk |v_\infty(k)|^2<\infty$ the model is UV-regular; $H_\infty$ is well defined on the domain $D(H_\infty)=D(H_0)$.
 \item If $\int \ud \bk |v_\infty(k)|^2=\infty$ but $\int \ud \bk \frac{|v_\infty(k)|^2}{k^2+1}<\infty$, the model remains UV regular, $H_\infty$ is well defined, but on a domain $D(H_\infty)$ different from $D(H_0)$ (in fact, $D(H_\infty)\cap D(H_0)=\{0\}$). This is the case for the original optical Fröhlich model, where $v(k)\propto 1/k$, $\omega\equiv\mathrm{const.}$ (see~\cite{GrWu16, LaSch19} for proofs and additional references). %
 \item If $\int \ud \bk \frac{|v_\infty(k)|^2}{k^2+1}=\infty$, the model is not UV-regular and the ground state energy of $H_\Lambda$ tends to negative infinity. However, if  
 \begin{equation}
  \int \ud \bk \frac{|v_\infty(k)|^2}{(k^2+1)^{2}}<\infty
 \end{equation}
 one can hope to define a renormalised Hamiltonian $H_\ren$ by simply renormalising the ground state energy.
 This was done rigorously by Nelson~\cite{nelson1964} for the case of a relativistic scalar field, i.e. $\omega(k)=\sqrt{k^2+m^2}$, $v(k)\propto \omega(k)^{-1/2}$, and later extended to a larger class of models (see e.g.~\cite{GrWu18, LaSch19}). For the class of models treated in those articles, it is sufficient to subtract the value of~\eqref{eq:E_0 pert} at $\bP=0$ from the Hamiltonian $H_\Lambda$ to obtain convergence (in the sense of resolvents)
 \begin{equation}
  H_\Lambda -E_\Lambda\uppar{1}(0) \rightarrow H_\ren.
 \end{equation}
 \item The Bogoliubov-Fröhlich Hamiltonian does not fall into the class treated in~\cite{nelson1964, GrWu18, LaSch19}. In view of the different behaviour of its ground state energy, it cannot be treated by the same methods. 
 In Sect.~\ref{sect:Hren} below we will construct the renormalised Hamiltonian for this case and prove convergence by a new method. This method was developed in~\cite{La19} for the case $v(k)\equiv\mathrm{const.}$ and $\omega(k)=k^2+1$ corresponding to a local interaction\footnote{In our model, the interaction has become non-local because of the Bogoliubov approximation}, and generalises the approach of~\cite{LaSch19}. 
 
 We believe our method can be extended to models satisfying $\int \ud \bk \frac{|v_\infty(k)|^2}{(k^2+1)^{2s}}<\infty$, for some $s<1$. For such a model we expect the divergence to show the behaviour
 \begin{equation}
  E_\Lambda \sim e_1 \Lambda^{\gamma_1} + \cdots + e_{j} \Lambda^{\gamma_j} + e_{j+1} \log \Lambda
 \end{equation}
for some integer $j=j(s)$, and real numbers $\gamma_1>\cdots>\gamma_j>0$ and $e_j=O(g^{2\ell})$. For arguments in support of 
this picture see Sect.~\ref{sect:pert}.

\end{enumerate}

\subsection{Construction of $H_\ren$}\label{sect:Hren}

Here we explain how to obtain the renormalised Hamiltonian $H_\ren$ and prove Theorem~\ref{thm:ren}. For some of the technical details we will refer to~\cite{La19}, where a similar result is proved for the model with $v_\infty(k)\equiv\mathrm{const.}$ and $\omega(k)=k^2+1$, which has the same UV-behaviour.

Let us start by describing the method for determining $H_\ren$. The key is to find a way of incorporating the diverging numbers $E_\Lambda$ into the full Hamiltonian that makes the expected cancellations explicit. 
The second-order contribution to the ground state energy can be written as
\begin{align}
 E_\Lambda\uppar{1}&= -g^2 \langle \emptyset | \Phi_\Lambda H_0(0)^{-1}  \Phi_\Lambda \emptyset \rangle \notag \\
 &= -g^2  \langle \emptyset | a(v_\Lambda) H_0(0)^{-1}a^\dagger(v_\Lambda)   \emptyset \rangle, \label{eq:E_0 pert det}\\
 &\sim-\frac{g^2\mu }{\pi^2}\Lambda
 ,\label{eq:E(1)}
\end{align}
with $a(v_\Lambda)=\int \ud \bk v_\Lambda(k) a_\bk$, and  reduced mass $\mu=(M^{-1} + c\xi \sqrt{2})^{-1}$. Set $a(v_\Lambda(\bx))=\int \ud \bk \ue^{\ui \bk \bx}v_\Lambda(k) a_\bk$ and 
\begin{equation}
 G_{\sigma,\Lambda}=-g (H_0+\sigma)^{-1}a^\dagger(v_\Lambda(\bx))
\end{equation}
(the shift $\sigma>0$ has been added to make $H_0$ strictly positive and thus invertible) . Then we can write 
\begin{align}
 H_\Lambda =& H_0 + g\left(a^\dagger(v_\Lambda(\bx)) + a(v_\Lambda(\bx))\right)\notag\\
 =& (H_0+\sigma)(1-G_{\sigma,\Lambda}) + ga(v_\Lambda(\bx)) -\sigma \notag\\
 = &(1-G^\dagger_{\sigma,\Lambda})(H_0+\sigma)(1-G_{\sigma,\Lambda}) - \sigma \label{eq:H_Lambda trafo1}\\
 &- g^2 a(v_\Lambda(\bx))(H_0+\sigma)^{-1} a^\dagger(v_\Lambda(\bx)). \label{eq:H_L Sigma}
\end{align}
Notice that~\eqref{eq:E_0 pert det} is just the ($\bP=0$, $\sigma=0$) vacuum expectation of~\eqref{eq:H_L Sigma}. To exhibit the cancellations, we spell out how the operator~\eqref{eq:H_L Sigma}, which preserves the boson-number, acts on a wavefuntion with a fixed number $n$ of bosons. We have
\begin{align}
 &-g^2a(v_\Lambda(\bx))(H_0+\sigma)^{-1} a^\dagger(v_\Lambda(\bx))\psi\uppar{n} \label{eq:Sigma^1 Lambda}\\
 &= \frac{-g^2}{n+1} \sum_{i,j=1}^{n+1} \int \ud \bk_i \frac{v_\Lambda(k_i)v_\Lambda(k_j)\psi\uppar{n}(\bp -\bk_i+\bk_j,  \cancel{\bk}_{j})}{H_0(\bp-\bk_i; \bk_1,\dots, \bk_{n+1})},\notag
\end{align}
where $\cancel{\bk}_{j}$ indicates that the function depends on all of the $\bk_\ell$ with $ \ell\neq j$.
Note that for $i=j$, $\psi\uppar{n}$ does not depend on the  variable of integration $\bk_i$ at all, so the operator given by this expression is just mulitplication by some function of $\bp,\cancel{\bk}_i$. By symmetry, it is equal to the function for $i=n+1$. 
This function diverges as $\Lambda\to \infty$, but combining with~\eqref{eq:E_0 pert det} yields
\begin{align}
 &\Sigma\uppar{1}_{\ud,\Lambda}(\bp, \bK)
 :=-g^2\int \ud \bq \left(\frac{|v_\Lambda(q)|^2}{H_0(\bp-\bq; \bK, \bq)+\sigma} - \frac{|v_\Lambda(q)|^2}{\frac{q^2}{2M} +\omega(q)}\right),
\end{align}
where $\bK=(\bk_1,\dots,\bk_n)$.
This integral is finite also for $\Lambda= \infty$, giving a function $\Sigma\uppar{1}_\ud=\Sigma_{\ud,\infty}(\bp,\bK)$. 

The terms in~\eqref{eq:Sigma^1 Lambda} with $i\neq j$ are quite different. They act on $\psi\uppar{n}$ as integral operators, and not multiplication operators. The result of this action will be finite, even for $\Lambda=\infty$, at least if $\psi\uppar{n}$ decreases sufficiently fast for large momenta. This operator is given by
\begin{align}
 &\left(\Sigma_{\mathrm{od},\Lambda}\uppar{1}\psi\uppar{n}\right)(\bp, \bK)
 =-g^2 \sum_{j=1}^{n} \int \ud \bq \frac{v_\Lambda(q)v_\Lambda(k_j)\psi\uppar{n}(\bp -\bq+\bk_j, \cancel{\bk}_j,\bq)}{H_0(\bp-\bq; \bK,\bq)+\sigma}.
\end{align}
We set $\Sigma\uppar{1}_\Lambda=\Sigma\uppar{1}_{\ud,\Lambda}+\Sigma_{\mathrm{od},\Lambda}\uppar{1}$, and $\Sigma\uppar{1}=\Sigma\uppar{1}_\infty$.
The operator $\Sigma\uppar{1}$ is related to the interaction between the bosons and the impurity mediated by exchange of a single boson.
The strength of this interaction is estimated in the following Lemma.

\begin{lemma}\label{lem:T}
There is an $n$-independent constant $C$ such that
\begin{equation}
 \| \Sigma\uppar{1} \psi\uppar{n}\| \leq g^2 C \| (H_0+\sigma)^{1/2}\psi\uppar{n}\|.
\end{equation}
Moreover, the difference to $\Sigma_{\Lambda}\uppar{1}$ satisfies for any $\eps>0$
\begin{equation}
 \| (\Sigma\uppar{1}_\Lambda- \Sigma\uppar{1})\psi\uppar{n}\| \leq C_\Lambda \| (H_0+\sigma)^{1/2+\eps}\psi\uppar{n}\|,
\end{equation}
with $\lim_{\Lambda \to \infty}C_\Lambda=0$.
\end{lemma}
\begin{proof}
 The estimates of the integral defining $\Sigma\uppar{1}_\ud$ follow by scaling and explicit evaluation (see~\cite[Sect.6]{FiTe2012}). The integral operator $\Sigma_\mathrm{od}\uppar{1}$ can be bounded by the Schur test. The important observation that the constant does not depend on $n$, even though the sum in $\Sigma_\mathrm{od}\uppar{1}$ contains $n$ terms, was made in~\cite{MoSe17}. A detailed proof is obtained by following the steps of~\cite[Lem.17]{La19}.
 
 The statement on convergence is obtained by applying the same reasoning with the form factor $v_\infty-v_\Lambda$ (see also~\cite{IBCrelat}).
\end{proof}

Since $\int \ud \bk \frac{|v_\infty(k)|^2}{(k^2+1)^2}<\infty$, the family of operators $G_{\sigma,\Lambda}$ also has a limit $G_{\sigma}:=G_{\sigma,\infty}$. More precisely:

\begin{lemma}\label{lem:G}
 For any $0\leq s<1/4$ there exists a constant $C$ and a family $C_\Lambda$ with $\lim_{\Lambda\to \infty} C_\Lambda =0$ such that for any $\Psi\in \cH$ we have
 \begin{equation}
  \| H_0^s G_\sigma \Psi\| \leq C \|\Psi\|, 
 \end{equation}
and
\begin{equation}
  \| H_0^s \left(G_\sigma-G_{\sigma,\Lambda}\right) \Psi\| \leq C_\Lambda \|\Psi\|.
 \end{equation}
 Moreover, there exists $\sigma_0\geq 0$ such that for all $\sigma>\sigma_0$, $1-G_\sigma$ is invertible with bounded inverse $\sum_{j=0}^\infty G_\sigma^j$.
\end{lemma}
\begin{proof}
 See~\cite[Cor.3.3]{IBCrelat}.
\end{proof}

Given these facts, we could hope to take the limit $\Lambda\to\infty$ of the expression~\eqref{eq:H_L Sigma} minus $E\uppar{1}_\Lambda$ together with~\eqref{eq:H_Lambda trafo1} to obtain $H_\ren$. In fact, this works for the less singular models from~\cite{nelson1964, GrWu18}, as demonstrated in~\cite{LaSch19}.
However, this does not work for the Bogoliubov-Fröhlich Hamiltonian, as can be inferred by again looking at the perturbative ground state energy (at $\bP=0$). Consider
\begin{align}
&H_\Lambda(0) - E\uppar{1}_\Lambda+\sigma
 =(1-G^\dagger_{\sigma,\Lambda})(H_0+\sigma)(1-G_{\sigma,\Lambda})\vert_{\bP=0} + \Sigma\uppar{1}_\Lambda\vert_{\bP=0}.
\end{align}
The ground state of the leading term is just $(1- G_{\sigma,\Lambda})^{-1} |\emptyset\rangle_{\bP=0}$.
The perturbative approximation of the ground state energy is thus
\begin{equation}
 \langle \emptyset | G_{\sigma,\Lambda}^\dagger \Sigma\uppar{1}_\Lambda G_{\sigma,\Lambda} \emptyset\rangle_{\bP=0} + g^5(\dots), \label{eq:GS pert2}
\end{equation}
since $\Sigma\uppar{1}\propto g^2$ and $G_{\sigma,\Lambda}\propto g$. As we have seen in Lemma~\ref{lem:T}, $\Sigma\uppar{1}_{\ud}\vert_{\bP=0}$ is a multiplier whose growth is proportional to $\sqrt{H_0(0)}\sim \frac{k}{\sqrt{2\mu}}$ on the one-boson space. On the other hand, $G_{\sigma,\Lambda}|\emptyset\rangle_{\bP=0}=-g (H_0(0)+\sigma)^{-1} v_\Lambda$, so the growth of~\eqref{eq:GS pert2} is proportional to
\begin{equation}
 %
  \frac{g^4}{\sqrt{2\mu}} \int \ud\bk \frac{|v_\Lambda(k)|^2 k}{(\frac{k^2}{2M}+\omega(k))^2} ,
\end{equation}
which diverges like $\log \Lambda$. The off-diagonal part $\Sigma\uppar{1}_\mathrm{od}$ will contribute at the same order. We set, choosing $\sigma=0$,
\begin{equation}\label{eq:E(2)}
 E\uppar{2}_\Lambda:=\langle \emptyset | G_{0,\Lambda}^\dagger \Sigma\uppar{1}_\Lambda G_{0,\Lambda} \emptyset\rangle_{\bP=0}.
 \end{equation}
By evaluation of the corresponding integrals we find
\begin{equation}
 E_2 \sim e_2 \log \Lambda ,
\end{equation}
with
\begin{align}\label{eq:e_2}
e_2=&\frac{g^4 \mu^3}{\pi^3} \gamma\left(\frac{\mu}{M}\right), \\
\gamma(s):=& \sqrt{1-s^2} - \frac1s \arctan\left(\frac{s}{\sqrt{1-s^2}}\right), \notag
\end{align}
and $\mu$ defined as in~\eqref{eq:E(1)}. This is in excellent agreement with the numerical results provided in~\cite[Fig.5.5]{grusdt2016c}.
From~\eqref{eq:e_2} we see that the contributions from $\Sigma\uppar{1}_\ud$ and $\Sigma\uppar{1}_\mathrm{od}$ cancel in the limit of a static impurity, i.e. $\lim_{M\to\infty} e_2(M)=0$ (see~\cite{IBCpaper} for a treatment of a static model). More precisely, we have the asymptotic behaviour for large $M$ 
\begin{equation}
 e_2= - \frac{2g^4\mu^5}{3\pi^3 M^2} +O(M^{-3})
\end{equation}
(compare\footnote{The formula~\cite[Eq.(5.43)]{grusdt2016c} differs from the large-$M$ approximation of $e_2(M)$ by the factor $\pi/4$. However, for the parameter $m/M$ chosen in~\cite[Fig.5.5]{grusdt2016c} the relative discrepancy of~\cite[Eq.(5.43)]{grusdt2016c} and $e_2$ is only approx.~$2\%$}~\cite[Eq.(5.43)]{grusdt2016c}).

We can extract this divergence from $H_\Lambda$ by a procedure similar to the first step, but treating $\Sigma\uppar{1}$ on the same footing as $H_0$.
We define
\begin{equation}
\tilde G_{\sigma,\Lambda}=-g \left(H_0+\Sigma\uppar{1}_\Lambda+\sigma\right)^{-1} a^\dagger(v_\Lambda(\bx)).
\end{equation}
By a calculation analogous to~\eqref{eq:H_Lambda trafo1}, we have
\begin{align}
 &H_\Lambda= 
 (1- \tilde G_{\sigma,\Lambda}^\dagger)(H_0+\Sigma\uppar{1}_\Lambda+\sigma)(1- \tilde G_{\sigma,\Lambda})
 -\sigma\\
 &- g^2 a(v_\Lambda(\bx))(H_0+\Sigma\uppar{1}_\Lambda+\sigma)^{-1} a^\dagger(v_\Lambda(\bx))-\Sigma\uppar{1}_\Lambda.
 \label{eq:H_L Sigma 2}
\end{align}
We now expand the last line using the resolvent identity and obtain
\begin{align}
 \eqref{eq:H_L Sigma 2}=&E_\Lambda\uppar{1} + G_{\sigma,\Lambda}^\dagger \Sigma\uppar{1}_\Lambda G_{\sigma,\Lambda}\label{eq:Sigma 2 sing}\\
  &- G_{\sigma,\Lambda}^\dagger  \Sigma\uppar{1}_\Lambda (H_0+\sigma)^{-1}  \Sigma\uppar{1}_\Lambda \tilde G_{\sigma,\Lambda} \label{eq:Sigma 2 reg}.
 \end{align}
From Lemma~\ref{lem:T} and Lemma~\ref{lem:G} we see that~\eqref{eq:Sigma 2 reg} defines a bounded operator, including for $\Lambda=\infty$, and that the expression for finite $\Lambda$ converges to the final one in the operator norm. 
We thus mainly need to analyse
\begin{equation}
 \Theta\uppar{2}_\Lambda:=G_{\sigma,\Lambda}^\dagger \Sigma\uppar{1}_\Lambda G_{\sigma,\Lambda} - E_\Lambda\uppar{2}.
\end{equation}

To this end, we consider, as in Eq.~\eqref{eq:Sigma^1 Lambda} before, how this operator will act on a wavefunction $\psi\uppar{n}$ with $n$ bosons.
Starting with the contribution due to $\Sigma\uppar{1}_{\ud,\Lambda}$, we observe again that some terms give multiplication operators by a function that diverges as $\Lambda\to \infty$, while others are integral operators and well defined for $\Lambda=\infty$, if $\psi\uppar{n}$ decreases quickly enough.
Spelling these terms out, we group the divergent terms together with $\langle \emptyset | G_{0,\Lambda}^\dagger \Sigma\uppar{1}_{\ud,\Lambda}G_{0,\Lambda}\emptyset\rangle_{\bP=0}$ and set (with $\bK=(\bk_1,\dots, \bk_{n})$ as before)
\begin{align}
&\Xi_{\ud,\Lambda}(\bp,\bK) 
= g^4 \int \ud \bq |v_\Lambda(q)|^2 \left(\frac{  \Sigma\uppar{1}_{\ud,\Lambda}(\bp-\bq; \bK,\bq)}{\left(H_0(\bp-\bq; \bK,\bq)+\sigma\right)^2}- \frac{\Sigma\uppar{1}_{\ud,\Lambda}(-\bq;\bq)}{H_0(\bq;\bq)^2} \right), \\\
&\left(\Xi_{\mathrm{od},\Lambda}\psi\uppar{n}\right)(\bp,\bK)\notag\\
&= g^4\sum_{j=1}^{n}
 \int \ud \bq \frac{v_\Lambda(q)v_\Lambda(k_j) \Sigma\uppar{1}_{\ud,\Lambda}(\bp-\bq, \bK, \bq) \psi\uppar{n}(\bp -\bq+\bk_j,  \cancel{\bk}_{j},\bq)}{\left(H_0(\bp-\bq; \bK,\bq)+\sigma\right)^2}.
\end{align}
By the same reasoning as for $\Sigma\uppar{1}$, both of these operators are well defined for $\Lambda=\infty$.

Applying the same procedure to $G_{\sigma,\Lambda}^\dagger \Sigma\uppar{1}_{\mathrm{od},\Lambda} G_{\sigma,\Lambda}$, we find

\begin{align}
&\Upsilon_{\ud,\Lambda}(\bp,\bK) = -g^4 \int  \ud \bq \ud \bq' \left( I(\bq, \bq', \bp, \bK, \sigma) - I( \bq, \bq', 0)\right), \label{eq:Yps d}\\
&I(\bq, \bq', \bp, \bK, \sigma)=
\begin{aligned}[t]
 &\frac{|v_\Lambda(q)|^2 |v_\Lambda(q')|^2}{\left(H_0(\bp-\bq;\bK,\bq)+\sigma\right)\left(H_0(\bp-\bq';\bK,\bq')+\sigma\right)}\\
&\times\frac{1}{\left(H_0(\bp-\bq-\bq';\bK,\bq,\bq')+\sigma\right)};
\end{aligned}\notag\\
&\left(\Upsilon_{\mathrm{od},\Lambda}\psi\uppar{n}\right)(\bp,\bK)
 = 
 - g^4\sum_{(j,\ell)\neq (n+1,n+2)}^{\substack{j\leq n+1\\ \ell\leq n+2}}
 \int \ud \bk_{n+1} \ud \bk_{n+2} I_{j,\ell}(\bp,\bk_1,\dots, \bk_{n+2}), \label{eq:Yps od}\\
 &I_{j,\ell}=
 \begin{aligned}[t]
 &\frac{v_\Lambda(k_{n+1})v_\Lambda(k_{n+2}) v_\Lambda(k_{j})v_\Lambda(k_{\ell})}
 {\left(H_0(\bp-\bk_{n+1}; \cancel{\bk}_{n+2})+\sigma\right)
 \left(H_0(\bp-\bk_{n+1}-\bk_{n+2}; \dots)+\sigma\right)
 }\\
 &\times \frac{\psi\uppar{n}(\bp -\bk_{n+1}-\bk_{n+2}+\bk_j+\bk_\ell,  \cancel{\bk}_{j},\cancel{\bk}_{\ell})}{\left(H_0(\bp-\bk_{n+1}-\bk_{n+2}+\bk_j; \cancel{\bk}_{n+1})+\sigma\right)}.
 \end{aligned}
\notag
\end{align}

Setting 
\begin{align}
%
 \Theta\uppar{2}_{\ud,\Lambda}:=&\Xi_{\ud,\Lambda}+\Upsilon_{\ud,\Lambda},\\
  \Theta\uppar{2}_{\mathrm{od},\Lambda}:=&\Xi_{\mathrm{od},\Lambda}+\Upsilon_{\mathrm{od},\Lambda},
\end{align}
we have, for finite or infinite $\Lambda$, 
\begin{align}\label{eq:Theta2}
  \Theta\uppar{2}_{\Lambda}=& \Theta\uppar{2}_{\ud,\Lambda} +  \Theta\uppar{2}_{\mathrm{od},\Lambda}.
  %
\end{align}
The operator $\Theta\uppar{2}_{\ud}$ is a multiplication operator by a function of logarithmic growth in $p, k_1,\dots,k_n$, while  $\Theta\uppar{2}_{\mathrm{od}}$ is an integral operator that is bounded.
In view of~\eqref{eq:Sigma 2 reg} we then set
\begin{equation}
 \Sigma\uppar{2}_{\Lambda}:=\Theta\uppar{2}_\Lambda - G_{\sigma,\Lambda}^\dagger  \Sigma\uppar{1}_\Lambda (H_0+\sigma)^{-1}  \Sigma\uppar{1}_\Lambda \tilde G_{\sigma,\Lambda},
\end{equation}
again omitting the subscript in the case $\Lambda=\infty$.

\begin{lemma}\label{lem:S}
 For any $\eps>0$ there is a constant $C$ and a family $C_\Lambda$ with $\lim_{\Lambda\to \infty} C_\Lambda =0$, such that
 \begin{equation}
  \|\Sigma\uppar{2} \Psi\|\leq g^4 C \|(H_0+\sigma)^\eps \Psi\|,
 \end{equation}
and
\begin{equation}
  \|(\Sigma\uppar{2} -\Sigma\uppar{2}_{\Lambda})\Psi\|\leq C_\Lambda \|(H_0+\sigma)^\eps \Psi\|.
 \end{equation}
\end{lemma}
\begin{proof}
As mentioned above, the regular part $\Sigma\uppar{2}-\Theta\uppar{2}$ defines a bounded operator by Lemma~\ref{lem:T} and Lemma~\ref{lem:G}.
 The bounds of the integral defining $\Theta\uppar{2}_{\ud}$ are elementary. The integral operator $\Theta\uppar{2}_{\mathrm{od}}$ is again bounded using the Schur test. A detailed argument showing that the norms do not grow with the boson-number is given in~\cite[Lem.19]{La19}.
\end{proof}

By the resolvent identity and Lemma~\ref{lem:T}, $\tilde G_{\sigma,\Lambda}$ and $\tilde G_{\sigma,\infty}=\tilde G_{\sigma}$ have essentially the same properties as $G_{\sigma,\Lambda}$.
\begin{lemma}\label{lem:tG}
 The statements of Lemma~\ref{lem:G} hold equally for $\tilde G_{\sigma,\Lambda}$.
\end{lemma}

In view of~\eqref{eq:H_L Sigma 2} we can now define the renormalised Hamiltonian by subtracting the appropriate constants, combining them with the operator~\eqref{eq:H_L Sigma 2} to form $\Sigma\uppar{2}$ and taking $\Lambda=\infty$. Explicitly, for $\sigma$ sufficiently large (as required for $1-\tilde G_{\sigma}$ to be invertible), we define
\begin{align}\label{eq:Hren}
 H_\ren=
 (1- \tilde G_{\sigma}^\dagger)(H_0+\Sigma\uppar{1}+\sigma)(1- \tilde G_{\sigma}) + \Sigma\uppar{2} -\sigma.
\end{align}
Note that $\Sigma\uppar{2}-\Theta\uppar{2}=\mathcal{O}(g^6)$ in operator norm, so for a perturbative treatment one may choose to neglect this, making the formulas completely explicit (see also Section~\ref{sect:pert}).

The domain of definition $D(H_\ren)$ consists of those vectors $\Psi\in \cH$ for which the first term is finite, i.e.,
\begin{equation}\label{eq:Dren}
 D(H_\ren):=\{ \Psi\in \cH | \|H_0(1-\tilde G_\sigma)\Psi\|<\infty \}.
\end{equation}

\begin{proposition}\label{prop:Hren}
 For any $\sigma> 0$ such that $(1-\tilde G_\sigma)$ is invertible, the operator $H_\ren$ with domain $D(H_\ren)$ is self-adjoint and bounded from below.
\end{proposition}
\begin{proof}
 Since $(1-\tilde G_\sigma)$ is invertible, the leading term of $H_\ren$ is clearly self-adjoint on $D(H_\ren)=(1-\tilde G_\sigma)^{-1}D(H_0)$ and non-negative. By the Kato-Rellich Theorem, it is thus sufficient to show that $\Sigma\uppar{2}$ is bounded relative to this term, with relative bound zero.
 Let $\Psi\in D(H_\ren)$. Then, using Lemma~\ref{lem:tG} and Lemma~\ref{lem:S}, we have for $\eps,\delta>0$ and some constants $C,C',C_\delta$
 \begin{align}
  \|\Sigma\uppar{2}\Psi\|
  \leq&C \|(H_0+\sigma)^\eps (1-\tilde G_\sigma)\Psi\| + \|\Sigma\uppar{2}\tilde G_\sigma\Psi\| \notag\\
  %
  %
  \leq & \delta \|H_0 (1-\tilde G_\sigma)\Psi\| + C_\delta \|\Psi\| \\
  \leq & C'\delta \|(1-\tilde G_\sigma^\dagger)H_0 (1-\tilde G_\sigma)\Psi\| + C_\delta \|\Psi\|,\notag
 \end{align}
where in the second step we also used Young's inequality. This proves the claim.
\end{proof}

To establish Theorem~\ref{thm:ren} it remains to prove that $H_\Lambda-E_\Lambda$ converges to $H_\ren$.
Note that this also proves that $H_\ren$ is independent of $\sigma>\sigma_0$, even though its definition explicitly uses $\sigma$.

\begin{proposition}\label{prop:conv}
Let $E\uppar{1}_\Lambda$ and $E\uppar{2}_\Lambda$ be given by~\eqref{eq:E_0 pert det}, respectively~\eqref{eq:E(2)} and $E_\Lambda=E\uppar{1}_\Lambda+E\uppar{2}_\Lambda$.
Then for any $\sigma>0$
 \begin{equation*}
  \lim_{\Lambda\to \infty} \left( H_\Lambda - E_\Lambda \pm\ui\sigma\right)^{-1} = \left(H_\ren\pm\ui\sigma\right)^{-1}
 \end{equation*}
 in the norm of operators on $\cH$.
\end{proposition}
\begin{proof}
Denote $\tilde H_\Lambda:=H_\Lambda-E\uppar{1}_\Lambda-E\uppar{2}_\Lambda$.
 By the identity~\eqref{eq:H_L Sigma 2} and the definition of $\Sigma\uppar{2}_\Lambda$ we have
 \begin{align}
 \tilde H_\Lambda=
 (1- \tilde G_{\sigma,\Lambda}^\dagger)(H_0+\Sigma\uppar{1}_\Lambda+\sigma)(1- \tilde G_{\sigma,\Lambda}) + \Sigma\uppar{2}_\Lambda -\sigma.
\end{align}
The difference of resolvents is then 
\begin{align}
 &(H_\ren\pm\ui\sigma)^{-1} - (\tilde H_\Lambda\pm\ui\sigma)^{-1} \\
 &= 
 \begin{aligned}[t]
     &(H_\ren\pm\ui\sigma)^{-1} 
     \bigg( (1- \tilde G_{\sigma}^\dagger)\left(\Sigma\uppar{1}_\Lambda-\Sigma\uppar{1}\right) (1- \tilde G_{\sigma}) \\
     & + \left((1- \tilde G_{\sigma}^\dagger)\left(H_0+\Sigma\uppar{1}_\Lambda\right)(\tilde G_{\sigma}- \tilde G_{\sigma,\Lambda})\right)\\
     &+\left((\tilde G_{\sigma}^\dagger- \tilde G_{\sigma,\Lambda}^\dagger)\left(H_0+\Sigma\uppar{1}_\Lambda\right)(1- \tilde G_{\sigma,\Lambda})\right) \\
     &+\left(\Sigma\uppar{2}_\Lambda-\Sigma\uppar{2}\right) 
     \bigg)(\tilde H_\Lambda\pm\ui\sigma)^{-1}.
\end{aligned}\notag
\end{align}
By the argument of Proposition~\ref{prop:Hren}, $(H_\ren\pm\ui\sigma)^{-1} (1- \tilde G_{\sigma}^\dagger)H_0$ and 
$H_0(1- \tilde G_{\sigma,\Lambda})(\tilde H_\Lambda\pm\ui\sigma)^{-1}$ define bounded operators, uniformly in $\Lambda$. By Lemma~\ref{lem:T}, $\Sigma\uppar{1}_\Lambda$, the first term then tends to zero in norm.
Since $\tilde G_{\sigma,\Lambda}\to \tilde G_{\sigma}$, also the second and third term converge to zero. From Proposition~\ref{prop:Hren} and Lemma~\ref{lem:tG} we see that $(H_\ren\pm\ui\sigma)^{-1}(H_0+1)^\eps$ defines a bounded operator for some $\eps>0$. Since $(H_0+\sigma)^{-\eps}\left(\Sigma\uppar{2}_\Lambda-\Sigma\uppar{2}\right)$ converges to zero in norm, by Lemma~\ref{lem:S}, this proves the claim.
\end{proof}

\subsection{Formal Perturbation Theory}\label{sect:pert}

Our technique of extracting the divergent contributions to the Hamiltonian can be regarded as a form of perturbative expansion. 
This expansion can be carried out to higher orders. In our case, these require no further renormalisation, but for more general models they lead to an iterative renomralisation procedure. This expansion could have practical importance for numerical simulations or the construction of trial states. Here, we make more explicit the general form of this expansion.
However, we remark that this expansion does not directly imply rigorous expansions of e.g. ground states of 
$H_\ren(\bP)$, since the model has no mass gap. 

Starting from the definition of $\Sigma\uppar{j}$, $j=1,2$, we can continue by recursively setting
\begin{align}
 T\uppar{j}&= H_0 + \sum_{i=1}^{j} \Sigma\uppar{i}  \\
 G\uppar{j}_\sigma&=  -g \left(T\uppar{j} + \sigma\right)^{-1} a^\dagger(v_\bx),
\end{align}
and then, using that $T\uppar{j}=T\uppar{j-1} + \Sigma\uppar{j}$ and the definition of $\Sigma\uppar{j}$,
\begin{align}
  &\Sigma\uppar{j+1}=-g^2 a(v_\bx)\left(T\uppar{j} + \sigma\right)^{-1}a^\dagger(v_\bx) - T\uppar{j}\notag+ H_0\\
  &= g^2 a(v_\bx) \left(T\uppar{j} + \sigma\right)^{-1} \Sigma\uppar{j} \left(T\uppar{j-1} + \sigma\right)^{-1}a^\dagger(v_\bx).
\end{align}
By the same calculation as in~\eqref{eq:H_L Sigma},~\eqref{eq:H_L Sigma 2}, we then have the identity,
\begin{equation}
 H_\ren = (1-G\uppar{j}_\sigma)^\dagger (T\uppar{j} + \sigma)(1-G\uppar{j}_\sigma)
 + \Sigma\uppar{j+1}-\sigma,
\end{equation}
 where we have already taken $\Lambda=\infty$ and $\sigma$ sufficiently large.
Note that $\Sigma\uppar{j}$ carries a power of $g^{2j}$. Additionally, by our bounds on $\Sigma\uppar{2}$, all $\Sigma\uppar{j}$ for $j>2$ are bounded operators on $\cH$, so we really have $\Sigma\uppar{j}=\mathcal{O}(g^{2j})$.
This formula clearly suggests to study $H_\ren - \Sigma\uppar{j+1}+ \sigma$, or its restriction to fixed momentum $\bP$, which is isospectral to $T\uppar{j}(\bP) + \sigma$.
These operators may have bound states with non-zero boson-number with molecule-like properties.

\subsection{The Model with Multiple Impurities}\label{sect:mult}

We now outline how the Hamiltonian for multiple impurities can be treated by the method presented above.
We treat the impurities as distinguishable particles, so that both fermions and bosons can accommodated by restricting to wavefunctions with the correct symmetry.

For a fixed number $N_\mathrm{I}$ of impurities, the Hamiltonian with cutoff is given by
\begin{equation}
 H_\Lambda:=H_0 + \sum_{j=1}^{N_\mathrm{I}} g \Phi_\Lambda(\bx_j),
\end{equation}
with
\begin{equation}
 H_0(\bp_1.\dots, \bp_{N_\mathrm{I}}; \bK)= \sum_{j=1}^{N_\mathrm{I}} \frac{\bp_j^2}{2M} + \ud \Gamma(\omega(\bk)).
\end{equation}
We now follow the procedure of Section~\ref{sect:Hren} and indicate the differences for the many-impurity case.
In general, now every creation or annihilation operator is associated with one of the impurities, at which the creation or annihilation takes place. 

To start with, we consider the generalisations of $\Sigma\uppar{1}$ and $E_\Lambda\uppar{1}$.
In Eq.~\eqref{eq:Sigma^1 Lambda} there is now an additional double sum over the impurities, with which both instances of $v_\Lambda$ are associated. When these impurities coincide, we obtain a contribution analogous to $\Sigma\uppar{1}+E_\Lambda\uppar{1}$. Terms corresponding to interaction of two different impurities give rise to additional integral operators in $\Sigma\uppar{1}$. These can be dealt with in the same way as before, see~\cite[Lem.7]{La19}.
Consequently, the numbers $E_\Lambda\uppar{1}$ should be replaced by $N_\mathrm{I} E_\Lambda\uppar{1}$. 
The interpretation of this is of course that the renormalisation adjusts the rest-energy of the impurities.

The next step is to consider $E_\Lambda\uppar{2}$ and $\Sigma\uppar{2}$. For this, it is necessary to understand the divergent terms in 
\begin{equation}
 G_{\sigma,\Lambda}^\dagger \Sigma\uppar{1}_\Lambda G_{\sigma,\Lambda}.
\end{equation}
As in the previous reasoning, they are characterised by the fact that they lead to multiplication operators, as opposed to integral operators. Note, however, that this expression now contains four sums over the impurities. Divergent terms arise whenever each of the indices coincides with another one, so one might expect that the correct choice of $E_\Lambda\uppar{2}$ behaves like $N_\mathrm{I}^2$.
This is not the case, due to cancellations between contributions coming from $\Sigma\uppar{1}_\ud$ and $\Sigma\uppar{1}_\mathrm{od}$, and the correct replacement for $E_\Lambda\uppar{2}$ is exactly $ N_\mathrm{I} E_\Lambda\uppar{2}$. 

This can be seen from the following calculations.
The relevant terms in $G_{\sigma,\Lambda}^\dagger \Sigma\uppar{1}_{\ud, \Lambda} G_{\sigma,\Lambda}$ are given by the evaluation at $\bp_1=\cdots=\bp_{N_\mathrm{I}}=0$ of 
\begin{align}
&-g^4\int \ud \bq' \frac{|v_\Lambda(q')|^2}{H_0(\bp_j - \bq',\dots; \bK, \bq') ^2} \\ 
 & \times\hspace{-3pt} \int\hspace{-3pt} \ud \bq  \left(\frac{|v_\Lambda(q)|^2}{H_0(\bp_i - \bq, \bp_j-\bq', \dots; \bK, \bq, \bq')} - \frac{|v_\Lambda(q)|^2}{H_0(\bq;\bq)}\right).\notag
\end{align}
The number of terms with $i=j$ is of course $N_\mathrm{I}$, and these give rise to exactly the same divergence as for $N_\mathrm{I}=1$.
The $N_\mathrm{I}(N_\mathrm{I}-1)$ contributions with $i\neq j$ each give
\begin{align}
&-g^4\hspace{-4pt}\int \hspace{-4pt}\ud \bq' \frac{|v_\Lambda(q')|^2}{H_0(\bq'
;\bq')^2}
 \hspace{-2pt}\int\hspace{-4pt}\ud \bq  \left(\frac{|v_\Lambda(q)|^2}{H_0(\bq,\bq';\bq,\bq')} - \frac{|v_\Lambda(q)|^2}{H_0(\bq;\bq)}\right) \notag\\
 &=g^4\hspace{-4pt}\int\hspace{-3pt} \ud \bq' \int\hspace{-3pt}\ud \bq \frac{|v_\Lambda(q')|^2|v_\Lambda(q)|^2 }{H_0(\bq';\bq')H_0(\bq,\bq';\bq,\bq')H_0(\bq;\bq) }.\label{eq:N_I cancel}
\end{align}
By inspection of the formulas~\eqref{eq:Yps d},~\eqref{eq:Yps od}, wee see that $G_{\sigma,\Lambda}^\dagger \Sigma\uppar{1}_{\mathrm{od}, \Lambda} G_{\sigma,\Lambda}$ gives exactly the same contribution as~\eqref{eq:N_I cancel}, but with the opposite sign, from the terms where the boson-indices $(j,\ell)=(n+1,n+2)$ are associated with different impurities.

\subsection{Interior-Boundary Conditions}\label{sect:IBC}

Here we explain briefly how the Bogoliubov-Fröhlich Hamiltonian, and in particular the condition~\eqref{eq:Dren} characterising its domain, can be understood in terms of interior boundary conditions.
More details on this approach can be found in~\cite{TeTu15, TeTu16}. It was used by the author in~\cite{La19} to study a three-dimensional model of Fröhlich type with point interactions, providing key insights for our treatment of the Bogoliubov-Fröhlich Hamiltonian.

A core idea of interior boundary conditions is that the configuration space of multiple particles should not contain configurations with more than one particle at the same location. It thus has a boundary, given by the configurations with $\bx=\by_i$, for some $i$ (or $\by_i=\by_j$ for $i\neq j$, but these play no role here since the bosons do not interact directly).
On this boundary one should then choose boundary conditions that correctly implement the physical model. For particles interacting via contact pseudo-potentials, with conserved particle number, these are of Skornyakov--Ter-Matirosyan type~\cite{skornyakov1956} (although these conditions are not sufficient in general,~\cite{FaMi1962}), or the Bethe-Peierls for the interaction with static particles.
In a model where the particle number is not conserved, the total boundary consists of the collision configurations of any number of particles. The interior boundary condition will then relate (generalised) boundary values of the wavefunctions with different particle numbers, thus giving rise to creation/annihilation of particles. 

Let us investigate the behaviour on the boundary for our model in more detail. 
This carries some important information, as it gives a criterion for a wavefunction to be an element of $D(H_\ren)$ that can be verified directly. In particular, any eigenfunctions of $H_\ren$ or the fixed-momentum operators $H_\ren(\bP)$ must satisfy these conditions (in the latter case with $\bx=0$).
Expressed on the sector with $n+1$ bosons, the condition~\eqref{eq:Dren} for $\Psi\in\cH$ to be in $D(H_\ren)$ implies
\begin{equation}
 \psi\uppar{n+1}- \tilde G_\sigma \psi\uppar{n} \in D(H_0)\cap \mathcal{H}\uppar{n+1}.
\end{equation}
In the position representation, every element of $D(H_0)\cap \mathcal{H}\uppar{n+1}$ 
can be evaluated on the plane $\{\by_j=\bx\}$ for any $j=1,\dots,n+1$ (by~\cite[Thm.IX.38]{ReSi2} -- this does not mean that the function is continuous), giving an element of $\cH\uppar{n}$.
In this sense, the condition means that $\psi\uppar{n+1} - \tilde G_\sigma \psi\uppar{n} $ is regular, and
thus $\psi\uppar{n}$ and $\tilde G_\sigma \psi \uppar{n}$ have the same singularities.
Spelling out $\tilde G_\sigma \psi \uppar{n}$, we find
\begin{align}
\frac{-g}{\sqrt{n+1}} \sum_{j=1}^{n+1}(H_0+\Sigma\uppar{1}+\sigma)^{-1} \check v(\bx-\by_j)\psi\uppar{n}(\bx,\cancel{\by}_j),
\end{align}
where $\check v$ is the inverse Fourier transform of $v_\infty(k)$. One easily checks, by comparing with the case $\check v= c \delta$, that the singularities of this expression are located on the planes where $\bx=\by_j$. Furthermore, the $j$-th term diverges only on this plane, but is regular at $\bx=\by_i\neq \by_j$, so the singularities at $\bx=\by_j$ are essentially the same as for $n=0$ and $\bx=\by$.
In that case, we have
\begin{align}
 &\tilde G_\sigma \psi \uppar{0}
 = - g \Big( (H_0+\sigma)^{-1} \check{v}(\bx-\by) \psi\uppar{0}(\bx)  \\
 &- (H_0+\sigma)^{-1}\Sigma\uppar{1}(H_0+\Sigma\uppar{1}+\sigma)^{-1} \check{v}(\bx-\by) \psi\uppar{0}(\bx)
 \Big).\notag
\end{align}
The term in the first line clearly has a divergence proportional to $|\bx-\by|^{-1}$ as $\by\to\bx$.
More careful analysis reveals that the second term has a divergence proportional to $\log|\bx-\by|$.
The function $\psi\uppar{n+1}$ thus has the asymptotic behavior
\begin{align}
 \psi\uppar{n+1}(\bx,\mathbf{Y})
  \sim \frac{-g}{\sqrt{n+1}} \left( c_{-1}|\bx-\by_j|^{-1})
 + c_0 \log|\bx-\by_j| \right)\psi\uppar{n}(\bx,\cancel{\by}_j), \label{eq:psi sing}
\end{align}
as $|\bx-\by_j|\to 0$, for all points with $\bx\neq \by_i$, $i\neq j$ and some constants $c_{-1}$, $c_0$ (their precise values  can be obtained from~\cite[Eq.(56)]{TeTu15} by choosing the parameter $m_y=\sqrt2 c \xi$).

In this context, the renormalisation procedure can be interpreted as the extension of the annihilation operator
$a(v_\bx)$ to functions with singularities as in~\eqref{eq:psi sing}. This extension is obtained by neglecting 
divergent contributions due to the ``evaluation'' of singular functions at $\by_{j}=\bx$ (cf.~\cite[Sect.3.2]{LaSch19} 
and~\cite[Eq.(29)]{La19}).
Let $A$ denote such an extension. Then we can express the action of the Hamiltonian $H_\ren$ on an element $\Psi\in D(H_\ren)$ of its domain as (cf.~\cite[Eq.(10)]{La19})
\begin{align}
 H_\ren\Psi= H_0 \Psi + a^\dagger(v_\bx)\Psi + A\Psi + E_0\Psi,
\end{align}
where $E_0$ is a constant that can be set to zero by changing the definition of $A$.
Note that in this equation every one of the first three summands should be interpreted as a distribution, and only their sum defines an element of $\cH$. This means that the domain of $H_\ren$ has been chosen specifically in order to make their singularities, which are located on the collision configurations, cancel each other out.

\section{Conclusions}

We have constructed the renormalised Bogoliubov-Fröhlich Hamiltonian by an explicit procedure. In doing so, we derived an exact expression for the constant of proportionality  of the $\log(\Lambda)$-energy-shift first observed in~\cite{grusdt2015}. The algebra underlying our method provides a way of expanding the Hamiltonian in powers of $g$ that is compatible with renormalisation.
We have also explained the relation to  the approach of interior boundary conditions, by which our method is inspired.
Our presentation of the Hamiltonian provides new tools for both analytical and numerical approaches to the BEC-polaron system. 

\bigskip
\noindent\textbf{Acknowledgment}
The author thanks Julian Schmidt and Fabian Grusdt for helpful discussions.

\bigskip
\noindent\textbf{Data Availability Statement}
Data sharing is not applicable to this article as no new data were created or analyzed in this study.

\end{document}